\documentclass[10pt,doublecolumn,journal]{IEEEtran}
\usepackage{amsthm} % For theorem
\usepackage{amsmath } % For maths 
\usepackage{mathtools}
\usepackage{graphicx}
\usepackage{caption}
\usepackage{subcaption}
\usepackage{color}
\usepackage{comment}
\usepackage{setspace} 
\usepackage{caption}
\usepackage{subcaption}
\usepackage{comment}
\newtheorem{theorem}{Theorem} % Theorem Counter reset
\newtheorem{lemma}{Lemma}% Lemma Counter retain the theorem counter
\begin{document}
\title{On Partial Coverage and Connectivity Relationship in Deterministic WSN Topologies }
 
\author{Rameshwar Nath Tripathi,
        Kumar Gaurav,
        and Y N Singh,~\IEEEmembership{Senior Member,~IEEE}% <-this % stops a space
}
\maketitle
\begin{abstract}
The primary function of a sensor network is to perform the sensing task. For proper sensing, the coverage and connectivity property of the network must be maintained. Relationship between coverage and connectivity has been extensively investigated for full coverage scenarios. But, investing in full coverage incurs unnecessary cost when partial coverage is sufficient.  In this work, we focus on the relationship between partial coverage and connectivity. We find the conditions when partial coverage implies connectivity as a function of coverage fraction. This condition facilitates the network designers to configure the connected network for desired coverage fraction.
\end{abstract}
\IEEEpeerreviewmaketitle
%============ INTRODUCTION============%
\section{Introduction}
A single node can sense only in a small region. Hence, to monitor large area, a network of sensor node is needed. If these sensor nodes can communicate  wirelessly, they are called Wireless Sensor Networks(WSNs). In the WSN, generally there will be a base station which collects measurements from all the nodes. As the nodes can communicate over short distances due to the limited capabilities, multi-hop communication is used for transferring measurements to base station. Thus, sensing as well as connectivity are the important features in a WSN. \\
Coverage of a sensor network represents how well the sensors monitor a field of interest (FoI) where they are deployed. It is the performance measure of the network sensing capability. Connectivity represents how well the nodes can communicate. It is the performance measure of the network communication capability. Characterization of these two measures help in the better design of sensor networks for different applications.
The coverage problem can be formulated as an optimization problem, i.e., to maximize lifetime subject to minimum coverage guarantee while maintaining the connectivity. \\
Coverage is broadly classified into three categories : (a) area coverage, (b) target coverage, and (c) barrier coverage. Area coverage concerns how well the sensor nodes monitor an area of interest (AOI). Target coverage requires monitoring of a set of targets instead of the whole sensing field. The problem of preventing an intruder from entering a Boundary of interest (BOI) is referred to as the barrier coverage. \\
The coverage characterizes the performance of a sensing network. It is represented by coverage fraction $\alpha$  where every point in the covered area = $\alpha$ x total area is sensed by at least one  sensor. Here, $0<\alpha \leq 1$. If coverage fraction $\alpha= 1$, then, we have full coverage. In the present paper, we focus on partial coverage  i.e, $ \alpha < 1$ . Since full coverage is stringent condition and expensive, it should be used only if necessary. Most of the applications \cite{yetgin2017survey} may work satisfactorily with partial coverage. For example, one such scenario is a sensor network for the weather forecast. Instead of knowing the humidity at every location in the field, measurement of  certain fraction of the area might prove sufficient for the humidity profile of the whole field. Therefore, applications may require  network configurations with different degree of coverage and connectivity. We can classify the connectivity and coverage (CC) requirements as follows: Connected full coverage (CFC) and  Connected partial coverage (CPC) \cite{li2011transforming}. \\
While solving the connected coverage (CC) optimization problem, the least possible number of constraints should be used to reduce the design complexity. The need for constraints reduction gives rise to two critical research questions concerning the design of wireless sensor networks:
\begin{enumerate}
\item  Are  coverage and connectivity constraints independent?
\item  If not, does coverage imply connectivity or vice versa so that a sensor network only needs to be configured to satisfy the stronger of the two conditions?
\end{enumerate}
Answering these questions qualitatively and quantitatively would greatly facilitate the design of WSNs. It will lead to insights into how to utilize the minimum number of nodes to achieve a desired coverage degree and coverage fraction while maintaining the required system connectivity. When number of node are more than needed, it will also tell us about increase in lifetime of WSN. \\
Most of the research work on sensing coverage in the literature has focused on the connected full coverage (CFC). Research in \cite{xing2005integrated, zhang2005maintaining} contributed to achieve full coverage as well as connectivity. However, the above two research questions have not been explored in the context of partial coverage. \\  
In this paper, we report an investigation on this relationship for deterministic deployment and the conditions which guarantee connectivity while configuring the desired coverage fraction. We discuss
 the network connectivity and coverage problem. We find the lower bound as well as the upper bound on communication radius and also derive an exact expression for communication radius in terms of coverage fraction. \\
 The  paper is organized as follows: in Section II, the notion of coverage and connectivity is
 introduced. Section III discusses the relationship between partial coverage and the connectivity problem in WSN. Finally, in Section IV, we conclude
 the outcome of this work along with the future research directions and  some open problems.   
 \vspace{-3.5mm} 
 \section{System Model} 
In this section, we present the network model along with its various components i.e. monitored space,  sensing region, coverage of a space point, Network coverage, node communication region and measure of network connectivity.\
The network consists of sensor nodes uniformly deployed sensor nodes over a rectangular field of interest (FoI), $\Psi$. Though, we have taken rectangular field, our results will hold true for any arbitrary field of interest. The events of interest (EOI), can take place randomly at any location in the sensing field. The events are represented as space points. A node can detect events within a certain distance, called sensing radius $R_S$. Let us consider a space point z and a set of sensors $\Omega = \{s_1, s_2, \ldots , s_N\}$ deployed in the sensing field. The \textit {d(s,z)} denotes the distance between a sensor $s$ and a space point $z$. Let $(s_x, s_y )$ and $(z_x, z_y )$ be the coordinates of the sensor $s$ and the space point $z$, respectively.\\
The coverage function of a space point $z$ relative to a sensor node \textit{i} is given by a binary value
\begin{align}
f(d(s_i,z)) = \left\{
        \begin{array}{ll}
            1 & \quad d(s_i,z) \leq R_S \\            
            0 & \quad  otherwise.
        \end{array}
    \right.
\end{align}
The value 1 implies that the sensor can monitor the point $z$.\\
A point might be covered by multiple sensors at the same time. The set of these multiple sensors is represented by $\Omega_Z\subset \Omega$. The coverage function of a space point relative to a set of sensors can be the addition of the coverage function of the point relative to each individual sensor. Thus, coverage function  $F(z)$ of a space point $z$ is defined as
\begin{align}
 F(z) = \sum\limits_{i=1}^{k} f(d(s_i,z)).
\end{align}
If $F(z) = k$, then we can say that the point $z$ is \textit{k}-covered.\\  
The coverage of the sensing field  relative to the deployed sensor nodes can be computed using equation (2), and is called the network coverage. The network coverage function $F_N$ is defined as the minimum value of $F(z)$ among all the possible values of $z$ in the entire network, i.e.,
\begin{align}
F_N = \displaystyle\min_{\forall z \in \Psi} F(z).
\end{align}
We assumes that each node $s_i$ is able to communicate only up to a certain distance from itself, called the communication radius $R_C$. 
The two nodes $s_i$ and $s_j$ are said to be connected if they are able to communicate directly with each other. In other words, the Euclidean distance between them must be less than or equal to their communication radius i.e. $d(s_i,s_j) \leq R_C$. 
The communication network of the set of nodes $S$ is modeled as a communication graph. A network is said to be connected iff there is at least a path between any pair of nodes as well as a path from every node to the base station.
%========= RELATIONSHIP BETWEEN PARTIAL COVERAGE===================%
\vspace{-2mm}
\section{RELATIONSHIP BETWEEN PARTIAL COVERAGE AND FULL CONNECTIVITY } 
In this section, we focus on understanding the relationship between partial coverage and full connectivity. We also find the conditions when partial coverage implies full connectivity. \\
There exist two possibilities-- (i) coverage implies connectivity and (ii) connectivity implies coverage. A connected network cannot guarantee coverage, because coverage is concerned with whether  every point of FoI is covered or not, While connectivity only requires all locations of active nodes to be connected. A covered network may guarantee connectivity, as coverage requires majority locations of the sensing field to be covered. The connectivity for a certain coverage depends upon the ratio of communication radius to sensing radius.\\ 
Zhang and Hou \cite{zhang2005maintaining} explored and established a connected full Coverage(CFC) condition. It states that -\\
\textit{  Assuming the monitored region is a convex set, the condition of $R_C\geq 2R_S$ is both necessary and sufficient to ensure that the complete coverage
of a convex region implies connectivity in an arbitrary network.}\\
The above condition is stated to be true for convex region, but we found that, this condition is also true for any arbitrary monitored region which may not be the convex region. The statement can be restated as:\\
\textit{For a given arbitrary monitored region, the condition of $R_C\geq 2R_S$ is both necessary and sufficient to ensure that complete coverage
of the region implies connectivity.}\\
The condition $R_C\geq 2R_S$ will also suffice as the connectivity criteria in contiguous partial coverage. In this paper, we will find a tighter bound on $R_C$ in terms of $\alpha$.\\
The goal in planning  a WSN is to maximize coverage with a minimum number of sensors. Thereby achieving least cost for the network. To realise this goal, the deployment process of the sensors plays an important role. A deployment is characterised by a deployment strategy and a deployment pattern. On the basis of the deployment strategy, the networks can be classified as structured or unstructured. In an structured sensor network, the nodes are placed at the specific locations, while in an unstructured sensor network, the nodes are placed randomly. Based on the area of the sensing field, area covered by deployed sensors, and overlapped sensing area among the nodes, the coverage can be classified as: exact coverage, over coverage and under coverage. The exact coverage means sensing field is covered without any overlap among the sensing regions of the nodes. If sensing field is covered and there exists some overlap among nodes, it is called over coverage. 
 To achieve exact coverage is not possible for circular sensing range. Only the possibilities of over-coverage and under-coverage exits. Formally, over coverage is called full coverage. Here, every point is covered by at least one sensor. Under coverage is called partial coverage where some points are not covered.
In a deterministic deployment to cover a plane, optimal placement pattern is  equilateral triangle lattice where sensor nodes are placed on the vertices of equilateral triangles.  \\
In this work, we study the partial coverage for the deterministic deployment strategies in the field of interest. We consider full coverage as the baseline to study partial coverage. Sensor nodes in triangular lattice pattern provides optimal full coverage when side length of the basic triangle pattern is $\sqrt{3}R_S$. This side length represents the maximum inter-node distance while ensuring 100\% coverage. If the inter-node distance is less than or equal to $\sqrt{3}R_S$, then coverage is full otherwise the FoI is under-covered. Therefore, partial coverage can only happen if $d\geq\sqrt{3}R_S$. If we keep on increasing the  spacing between nodes, overlap area just becomes zero at $d=2R_S$. For $d\geq 2R_S$, there is no overlap area. Therefore, for study of partial area coverage, the region of interest is $d\geq\sqrt{3}R_S$. 
we can divide this region into two sub-regions, i.e,$ \sqrt{3}R_S\leq d\leq 2R_S$ and $d\geq 2R_S$. We compute the spacing between nodes in triangular mesh to achieve desired $\alpha$. We refer to this spacing $d_\alpha$. To make our discussion precise, we will derive $d_\alpha$.

%============ Lemma1=====================
 \begin{lemma}
 Given that the nodes with sensing radius $R_S$ are deployed in the triangular mesh configuration in the sensing field $\Psi$, the coverage fraction of the sensing field varies non-linearly with spacing between two adjacent nodes, in the range $\sqrt{3}R_S\leq d\leq 2R_S$ 
\end{lemma}
 \begin{proof}
 The nodes are deployed in a rectangular region of size ${L\times L}$. Consider the simple scenario, where any two nodes $s_i$ and $s_j$ are neighbour and $s_j$ lies on either side of $s_i$.  The inter-node distance between neighbouring nodes $s_i$ and $s_j$ is given by $d=d(s_i, s_j)$. With the large number of nodes, the sensor placement is as shown in Fig. \ref {fig:RDS}. It is just a regular pattern of the configuration as shown in Fig. \ref {fig:TP}.  
The common sensing area of two nodes is denoted by $I_{ij}$. It is given by
\begin{align*}
 I_{ij} = R_S^2\Bigg[2\arccos{\bigg(\frac{d}{2R_S}\bigg)}-\frac{d}{R_S}\sqrt{1-\bigg(\frac{d}{2R_S}\bigg)^2}\Bigg].
  \end{align*}
 The overall coverage ratio, $\alpha$, of the sensing field $\Psi$ in  Fig. \ref{fig:RDS} will be equal to the coverage ratio of triangular pattern in Fig. \ref{fig:TP}. 
\begin{align*}
 \alpha = \frac{R_S^2\Big[\frac{\pi}{2}-3\arccos{\big(\frac{d}{2R_S}\big)}+\frac{3d}{2R_S}\sqrt{1-\big(\frac{d}{2R_S}\big)^2}\Big]}{\frac{\sqrt{3}}{4}d^2}.
  \end{align*}
  We define inter-node distance relative to $2R_S$ as $\beta=d/2R_S$, and putting this value in $\alpha$,we get
  \begin{align} \label{eq:4}
 \alpha = \frac{1}{\sqrt{3}\beta^2}\Bigg[\frac{\pi}{2}-3\arccos{\big(\beta\big)}+3\beta\sqrt{1-\big(\beta\big)^2}\Bigg].
  \end{align}
\end{proof}
\vspace{-4mm}
%======END Lemma 1 ==========================%
 Equation (4) represents the relationship between $\alpha$ and node spacing $d$ in the triangular mesh. We can set the spacing between triangle vertices to meet desired $\alpha$-coverage. We denote this spacing by $d_\alpha$. Thus, for a given $\alpha$, one can compute the $d_\alpha$. For $\alpha = 0.906$, the $d_\alpha$ is $2R_S$. The $d_\alpha =2R_S$ represents  a special configuration where the overlap among adjacent nodes just becomes zero. We refer to this $d_\alpha$ as $d_p$ which denotes the point overlap. The mere knowledge of relative inter node distance with respect to $d_p$ can be used to define bounds on $\alpha$. For given arbitrary $d$, bounds on $\alpha$ can be given as.
 \begin{align}
\alpha = \left\{
        \begin{array}{ll}
            < 90.6\%& \quad    d >d_p  \\ 
              90.6\% & \quad   d=d_p \\           
             > 90.6\% & \quad   d<d_p .
        \end{array}
    \right.
\end{align}
 
%============Lemma2=================================
\begin{figure*}[ht!]
\hspace{1cm}
\centering
\begin{subfigure}[h] {0.4\textwidth}
\centering
\includegraphics[scale =0.3]{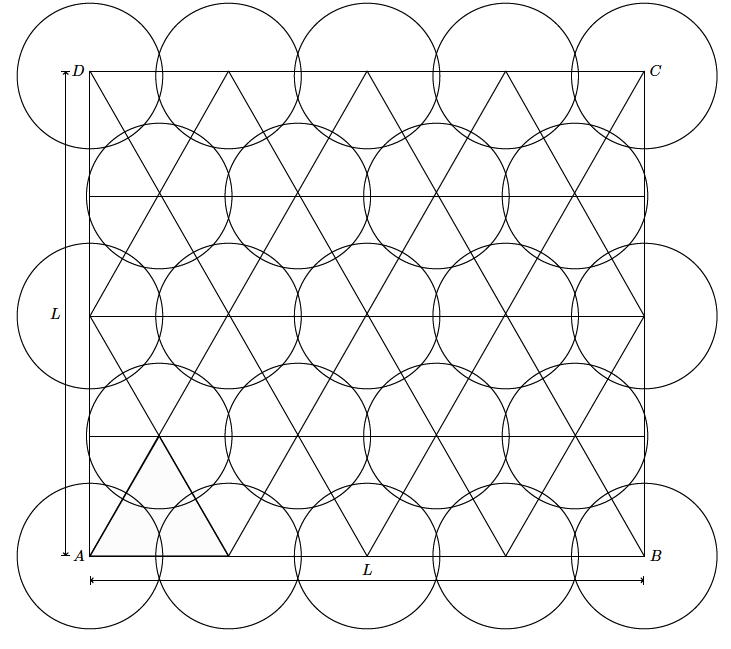}
\caption{}
\label{fig:RDS}
\end{subfigure}
\hspace{1cm}
\begin{subfigure}[h]{0.4\textwidth}
\centering
\includegraphics[scale =0.4]{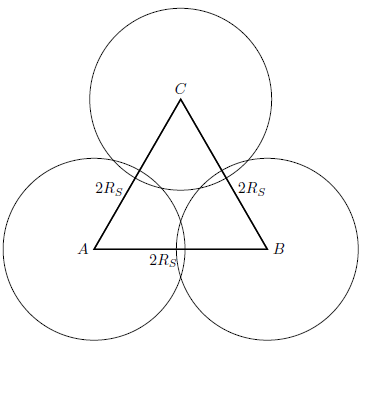}
\caption{}
\label{fig:TP}
\end{subfigure}
\caption{Triangular deployment for maximum coverage with nodes' sensing region overlap in the sensing field: (a) Deployment in sensing field, and (b) Basic unit of the pattern }
\label{fig:Triangular Deployment}
 \end{figure*}
\vspace{-2mm}
\begin{lemma}
The coverage ratio of a sensing field is inversely proportional to the square of inter-node distance $d$ when $d\geq 2R_S$. 
\end{lemma}
\begin{proof}
From Lemma 1, one can compute $d_\alpha$ for coverage ratio $\alpha$, where $\alpha <0.9068$. 
The nodes are placed  as shown in Fig.\ref{fig:FT}. One can observe that the pattern in Fig. \ref{fig:FSD} can be repeated to get the configuration shown in Fig. \ref{fig:FT}. The  $\alpha$ for the sensing field, $\Psi$ is equal to
 \begin{figure*}[ht!]
\hspace{1cm}
\centering
\begin{subfigure}[h] {0.4\textwidth}
\centering
\includegraphics[scale =0.3]{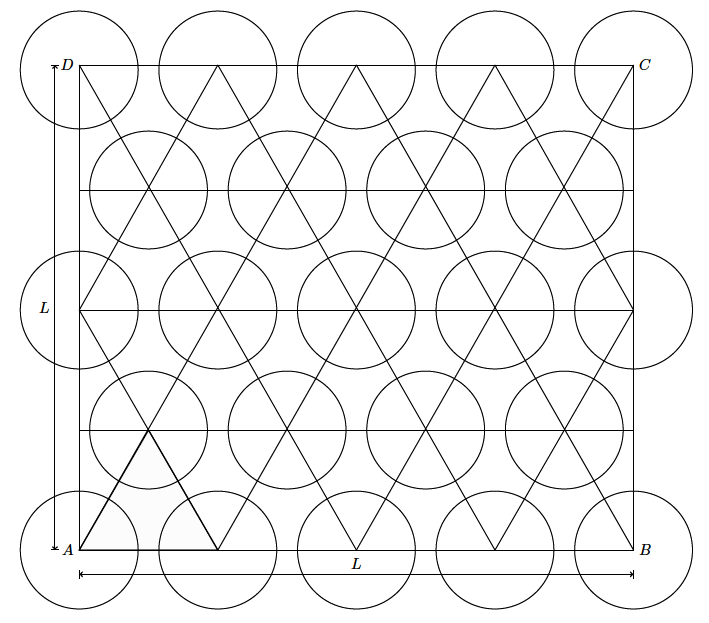}
\caption{}
\label{fig:FT}
\end{subfigure}
\hspace{1cm}
\begin{subfigure}[h]{0.4\textwidth}
\centering
\includegraphics[scale =0.5]{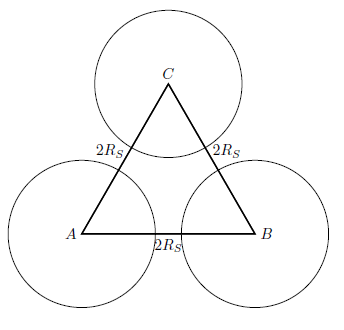}
\caption{}
\label{fig:FSD}
\end{subfigure}
 \caption{Triangular deployment for arbitrary partial coverage without any overlap in nodes' sensing regions : (a) Deployment in sensing field, and and (b) Basic element in the deployment }
\label{fig:fig_image2}
\end{figure*}
 \begin{align} \label{eq:5}
 \alpha &= \frac{\text {Area of circular region in $\bigtriangleup$ABC}}{\text{Area of $\bigtriangleup$ ABC}}\nonumber\\
&= \frac{2 \pi {R_s}^2} {\sqrt{3}d^2}. 
\end{align}
Thus, for a given sensing radius, the coverage ratio  depends only on the spacing between two adjacent nodes and follows the inverse square law. 
\end{proof}
\vspace{-3mm}
In this case also, we can set the spacing between triangle vertices to meet desired $\alpha$-coverage.
 So far in this letter, we have focused on the coverage problem, which ensures that an event happening at any point in FoI is detected. In order to transfer sensed information network must be connected. The following theorem provides the condition on the communication range to ensure connected network for the required $\alpha$.

%============Theorem1=========== 
\begin{theorem}
For a given coverage fraction, $\alpha \in (0.906,1)$, we will have ${R^{min}_C} \in (\sqrt{3} R_S, 2 R_S)$. Having $R_C \geq {R^{min}_C}$  is both  necessary and sufficient condition to ensure that $\alpha$-coverage implies connectivity.  
\end{theorem}
\begin{proof}
For the considered deterministic deployment, the sensing network is connected if minimum communication range of nodes is greater than or equal to the maximum separation between any two adjacent nodes for a given $\alpha$ i.e., ${R^{min}_C} \geq d_\alpha$.
For an $\alpha \in (0.906,1)$, corresponding $d_\alpha$ can be obtained in two ways- direct computation using equation (\ref{eq:4}), and retrieval of data-pair $(\alpha,d_\alpha)$ from a look up table.
 The inter-node distances corresponding to 0.906 and 1 are $2R_S$ and $\sqrt{3}R_S$ respectively. Therefore for $\alpha \in (0.906,1)$, the ${R^{min}_C}$ must be in the range $(\sqrt{3}R_S,2R_S)$.
\end{proof}
%==================Theorem2===== \vspace{-4mm}
\begin{theorem}
For a given $\alpha$ in the range $\alpha\leq 0.906 $, the condition $D_{SF} \leq {R^{min}_C}\leq (\sqrt{\frac{2\pi}{\sqrt{3}\alpha}})\times{R_S}$ is both necessary and sufficient to ensure that  $\alpha$-coverage implies connectivity. Here, $D_{SF}$ is sensing field diameter.  
\end{theorem}
\vspace{-4mm}
\begin{proof}
To maintain connectivity in the any sensing field, communication radius $R_C$ must be $\geq d$. We can compute the value of  $d$ from equation (\ref{eq:5}) to get $R_C$ as a function of $\alpha$ and $R_S$, i.e
\vspace{-2mm}
\begin{align}
R_C\geq (\sqrt{\frac{2\pi}{\sqrt{3}\alpha}})\times{R_S}.
\end{align}
Above expression gives the lower bound on $R_C$. To get upper bound on the required communication radius, we consider the worst case scenario of a two-node network. To create the worst-case scenario, we place the nodes at the farthest possible points in the sensing field. In this scenario, the distance between two nodes is equal to the sensing field diameter. The diameter of a shape is the upper bound on the set of all pairwise distances between all points in the area $\phi$. It is given by
\begin{align}
D_{SF} = sup\{ d(z_1, z_2) \mid z_1,z_2 \in \phi\}.
\end{align}
In other words, maximum inter-node distance is $D_{SF}$. To form a connected network, these nodes would need to set their minimum communication radius equal to the sensing field diameter. 
 The communication radius need not to be set greater than the diameter of the sensing field  as it gives no additional advantage and leads to waste of energy. 
 \end{proof}
%\vspace{-2mm} 
 Although the relationship ${R^{min}_C} \geq 2R_S$ is necessary to guarantee network connectivity provided that the required coverage is ensured. We have computed the lower and the upper bounds on the minimum communication radius for the various ranges of $\alpha$ as 
 \begin{align}
{R^{min}_C} = \left\{
        \begin{array}{ll}
   f(\alpha)R_S & \quad   \alpha <  90.6\%  \\ 
              2R_S &\quad   \alpha= 90.6\% \\           
             \in(\sqrt{3}R_S,2R_S) & \quad  \alpha \in (90.6 \%, 100\%).
        \end{array}
    \right.
\end{align}
Here,$ f(\alpha)=\sqrt{\frac{2\pi}{\sqrt{3}\alpha}}$ and for $\alpha \in (90.6 \%, 100\%)$ the exact value of ${R^{min}_C}$ can be computed using equation (4), as an implicit function of $d_\alpha$. 
These bounds help to reduce the  communication energy. This could be exploited by the designers to conserve the energy in the network.
\vspace{-2mm}
\section{Conclusion}
This letter explored the problem of maintaining both the desired partial coverage and full connectivity in deterministic-ally deployed wireless sensor networks. We proved that $\alpha$-coverage can imply full network connectivity  and computed the tighter lower  and upper bounds on the minimum communication radius for various ranges of $\alpha$. At a high-level, our analysis of $\alpha$-coverage  advocates  the use of our results because of the potential energy savings. The work in this paper can be extended in several directions. The first extension is to find the exact or approximate partial coverage condition for the random deployment. Second possible extension is to consider probabilistic sensing model and probabilistic communication model and third extension is to use results in the design of topology control  for WSNs.
\vspace{-3mm}
%============END-CONCLUSION========================
\ifCLASSOPTIONcaptionsoff
  \newpage
\fi
 
\bibliography{Reference}

\begin{thebibliography}{1}

\bibitem{yetgin2017survey}
Halil Yetgin, Kent Tsz~Kan Cheung, Mohammed El-Hajjar, and Lajos~Hanzo Hanzo.
\newblock A survey of network lifetime maximization techniques in wireless
  sensor networks.
\newblock {\em IEEE Communications Surveys \& Tutorials}, 19(2):828--854, 2017.

\bibitem{li2011transforming}
Yingshu Li, Chinh Vu, Chunyu Ai, Guantao Chen, and Yi~Zhao.
\newblock Transforming complete coverage algorithms to partial coverage
  algorithms for wireless sensor networks.
\newblock {\em IEEE Transactions on Parallel and Distributed Systems},
  22(4):695--703, 2011.

\bibitem{xing2005integrated}
Guoliang Xing, Xiaorui Wang, Yuanfang Zhang, Chenyang Lu, Robert Pless, and
  Christopher Gill.
\newblock Integrated coverage and connectivity configuration for energy
  conservation in sensor networks.
\newblock {\em ACM Transactions on Sensor Networks (TOSN)}, 1(1):36--72, 2005.

\bibitem{zhang2005maintaining}
Honghai Zhang and Jennifer~C Hou.
\newblock Maintaining sensing coverage and connectivity in large sensor
  networks.
\newblock {\em Ad Hoc and Wireless Sensor Networks}, 1(1-2):89--124, 2005.

\end{thebibliography}
\bibliographystyle{unsrt}

\end{document}